\documentclass[11pt]{article}

\usepackage{amsfonts}
\usepackage{amssymb,amsmath}
\usepackage{amsthm,amsgen}
\usepackage{graphicx,epsfig}
\usepackage[mathscr]{eucal}
\usepackage{slashed}

\usepackage[usenames,dvipsnames]{color}
\usepackage[usenames,dvipsnames]{colortbl}

\newcounter{subequation}
	\newenvironment{subequation}%
	{\addtocounter{equation}{-1}%
	\stepcounter{subequation}%
	\begin{equation}}%
	{\end{equation}%
}

\newcommand{\ts}{\textstyle}

\newcommand{\bA}{\mathbf{A}}

\newcommand{\bb}{\mathbf{b}}

\newcommand{\be}{\mathbf{e}}
\newcommand{\bF}{\mathbf{F}}

\newcommand{\bm}{\mathbf{m}}
\newcommand{\calm}{\mathbf{m}}

\newcommand{\bR}{\mathbf{R}}

\newcommand{\bT}{\mathbf{T}}

\newcommand{\bg}{\mathbf{g}}

\newcommand{\beq}{\begin{equation}}
\newcommand{\eeq}{\end{equation}}
\newcommand{\bseq}{\begin{subequation}}
\newcommand{\eseq}{\end{subequation}}
\newcommand{\refeq}[1]{(\ref{#1})}

\newcommand{\fM}{\mathfrak{M}}
\newcommand{\fN}{\mathfrak{N}}

\newcommand{\fS}{\mathfrak{S}}
\newcommand{\p}{\partial}

\newcommand{\cA}{\mathcal{A}}

\newcommand{\cF}{\Phi}
\newcommand{\cU}{\mathcal{U}}
\newcommand{\cR}{\mathcal{R}}

\newcommand{\cD}{\mathcal{D}}
\newcommand{\cE}{{\mathcal E}}

\newcommand{\cG}{{\mathcal G}}

\newcommand{\cL}{{\mathcal L}}
\newcommand{\cM}{{\mathcal M}}
\newcommand{\cN}{{\mathcal N}}

\newcommand{\cQ}{{\mathcal Q}}
\newcommand{\cS}{{\mathcal S}}
\newcommand{\cT}{\mathcal{T}}

\newcommand{\M}{\mathcal{M}}

\newcommand{\Cset}{{\mathbb C}}

\newcommand{\Rset}{{\mathbb R}}
\newcommand{\Sset}{{\mathbb S}}
\newcommand{\Zset}{{\mathbb Z}}

\newcommand{\la}{\lambda}

\newcommand{\Del}{\Delta}
\newcommand{\al}{\alpha}

\newcommand{\ga}{\gamma}
\newcommand{\Ga}{\Gamma}
\newcommand{\ep}{\epsilon}

\newcommand{\ka}{\kappa}

\newcommand{\Om}{\Omega}

\newcommand{\Si}{\Sigma}
\newcommand{\nab}{\nabla}
\newcommand{\half}{\frac{1}{2}}

\newcommand{\bna}{\begin{eqnarray}}
\newcommand{\ena}{\end{eqnarray}}
\newcommand{\bea}{\begin{eqnarray*}}
\newcommand{\eea}{\end{eqnarray*}}
\newcommand{\ben}{\begin{enumerate}}
\newcommand{\een}{\end{enumerate}}
\newcommand{\bi}{\begin{itemize}}
\newcommand{\ei}{\end{itemize}}

\newcommand{\N}{{\mathcal N}}

\newcommand{\RR}{{\mathbb R}}

\newcommand{\Hset}{{\mathbb H}}

\newcommand{\sm}{\textsc{m}}
\newcommand{\scp}{\textsc{p}}
\newcommand{\sq}{\textsc{q}}
\newcommand{\zM}{\cM_0}
\newcommand{\zg}{\bg_0}
\newcommand{\zF}{\bF_0}
\newcommand{\zfM}{\fM_0}
\newtheorem{thm}{THEOREM}[section]
\newtheorem*{thm*}{THEOREM}

\theoremstyle{definition}

\theoremstyle{remark}

\begin{document}
%
%

\title{On {a} zero-gravity limit of the Kerr--Newman \\ spacetimes and their electromagnetic fields}
\author{A. Shadi Tahvildar-Zadeh\footnote{Department of Mathematics, Rutgers, The State University of New Jersey, 
110 Frelinghuysen Rd., Piscataway, NJ 08854}}

\date{October 1, 2014; revised 11/14/14}

\maketitle

\begin{abstract}

\noindent
We discuss the limit of vanishing $G$ (Newton's constant of universal gravitation) of the {maximal 
analytically extended} Kerr--Newman electrovacuum spacetimes {represented in Boyer--Lindquist coordinates}.
We investigate the topologically nontrivial spacetime $\M_0$ emerging in this limit and show that it
consists of two copies of flat Minkowski spacetime cross-linked at a timelike solid
cylinder (spacelike 2-disk $\times$ timelike $\RR$).
As $G\to 0$, the electromagnetic fields of the Kerr--Newman spacetimes converge to nontrivial solutions of
Maxwell's equations on this background spacetime $\M_0$.
We show how to obtain these fields by solving Maxwell's equations with singular sources
supported only on a circle in a spacelike slice of $\M_0$.  
These sources do not suffer from any of the pathologies that plague the alternate sources found
in previous attempts to interpret the Kerr--Newman fields on the topologically simple Minkowski spacetime. We characterize the singular behavior of these sources and  prove that the Kerr--Newman electrostatic potential and magnetic stream function are the unique solutions of the Maxwell equations among all functions that have the same blow-up behavior at the ring singularity.
\end{abstract}

\section{Introduction}

The Kerr--Newman (KN) electromagnetic spacetimes \cite{NCCEPT65} are a three-parameter
family of triplets $(\M,\bg,\bF)$ consisting of a four-dimensional manifold $\M$ endowed
with a Lorentzian metric $\bg$ and an {\em electromagnetic field}, i.e. an exact 2-form $\bF=d\bA$.
The manifold $(\cM,\bg)$ is asymptotically flat and possesses two commuting Killing fields in such
a way that, in a neighborhood of spatial infinity, one is timelike with $\RR$-orbits
(parametrized by $t$) and the other spacelike with $\Sset^1$-orbits (parametrized by $\varphi$). Their ``outer regions" thus belong to the class of {\em stationary, axisymmetric} spacetimes studied by Lewis \cite{Lew32} and Papapetrou \cite{Pap53}.
The KN electromagnetic field is invariant under the flow of both Killing fields,
and in an asymptotically flat neighborhood of spatial infinity displays,
in leading order, {a structure corresponding to}
an electric monopole and a magnetic dipole {field in overall flat space.}
This region of the KN spacetimes is therefore thought to model the {\em electromagnetic vacuum}
exterior to an axisymmetric, stationarily spinning, charged astrophysical body.

 However, a regular continuation
into the {\em interior} of such an object, for some {\em non-exotic}\footnote{By {\em non-exotic}
matter we mean, to paraphrase H. Bondi (as quoted in \cite{Bic93,Kle03}), matter that ``can be bought
in the  shops.'' Thus, in particular, something infinitely charged or massive, or violating positive
energy conditions, or moving at superluminal speeds with respect to infinity, would be considered exotic.}
matter model, has not been achieved.  In fact it is not known whether this is possible at all.

On the other hand, these stationary ``exterior'' KN spacetimes do have a maximal {\em analytic} extension which
ends at a spacetime-curvature singularity.
Already Newman {\em et al} \cite{NCCEPT65} tentatively identified this singularity as being a
``rotating ring of charge,'' though a conspicuous footnote in the paper \cite{NewJan65}
indicates that they were made aware of the difficulties with this interpretation by the referee.
Carter \cite{Car68} then undertook a thorough investigation of the maximal analytic extension of
the KN spacetimes, which revealed that all metrics in this family are indeed singular on a
cylindrical surface whose cross-section at fixed $t$ is a circle.
However, he also showed that this constant-$t$ ``ring-singularity'' is {\em timelike}.
Worse, he showed that the Killing field $\p/\p \varphi$ generating the $\Sset^1$ orbits becomes
timelike already near the singular cylinder, so that the KN spacetime has a noncausal zone of
{\em closed timelike loops} in the vicinity of its singular cylinder.
For some range of the three KN parameters (usually referred to as charge, mass, and (spin-)angular
momentum of the spacetime) this noncausal zone is hidden behind an event horizon, residing inside a
``black hole'' region of the spacetime, but for another range of parameter values it's not, and
the ring singularity is then naked.

Carter's thorough study \cite{Car68} of the singularities of the maximal analytic
extension of the KN spacetimes seems to rule out the possibility of referring to the singular set of these
spacetimes in any conventional way as a ``stationary, spinning charge distribution.'' Curiously however, it has not put to
rest the quest for such a source, see \cite{Isr70,Tio73,Lop83, PekFra87,Kai04,Lyn04}.

In this paper we present a different angle of attack on this old ``problem of the sources for KN" by studying
what is quite possibly {\em the only situation} in which talking about ``stationarily spinning sources"
for a KN-type manifold is feasible,
namely, we discuss {a} {\em  zero-gravity limit} of the KN spacetimes.
We find that in this limit a flat, but topologically nontrivial, spacetime $\M_0$ emerges which consists
of two copies of Minkowski spacetime cross-linked at a timelike ``solid'' cylinder, a 2-disk $\times$ $\RR$.  

Such two-sheeted spacetimes where first discovered by D. M. Zipoy \cite{Zip66}, who found a two-parameter 
family of static axisymmetric solutions to Einstein's vacuum equations with this {topology}. 
The spacetime $\M_0$ that we study in this paper is {also a $G\to0$ limit of}
the Zipoy family.

The same spacetime manifold can also be obtained by taking  {a} limit
in which both the charge and the mass of the KN solution vanish, as was already contemplated by Carter \cite{Car68}.
However, in contrast with Carter's limit, in our limit the electromagnetic fields of the KN spacetime
converge to {\em nontrivial} solutions of
the familiar vacuum Maxwell equations  on this flat, topologically nontrivial background $\M_0$.
We will show that these fields can be interpreted as stationary solutions of  {the} 
Maxwell equations with
singular sources prescribed on a time-slice of $\M_0$. 
The crucial difference between our approach and the previous  
attempts to identify singular sources for KN fields is thus threefold:
\begin{enumerate}
\item
Our sources (charge and current) are supported {\em only} on the ring.
\item The ring singularity is {\em spacelike} in the limit $G\to 0$, thus making it possible to interpret the source as a stationary, spinning charge distribution.
\item Even though the underlying spacetime is flat, it is topologically nontrivial{:} it is {\em double-sheeted}, and the electromagnetic KN potentials are {\em anti-symmetric} functions defined on it.
\end{enumerate}
In particular, the third observation in the above, and thus the connection with Zipoy {topology}, 
is lacking in all previous treatments  of the zero-gravity limit of KN spacetimes (e.g. \cite{Car68,New02,GaiLyn07}) 
and the KN electromagnetic fields (e.g. \cite{Isr70,Tio73,Lop83, PekFra87,Kai04,Lyn04a,Lyn04}) that we are aware of.

The rest of this paper is structured as follows:

In the next section we collect the mathematical preliminaries needed for our discussion of
the KN spacetimes;
we first recall the conventional Einstein--Maxwell\footnote{One should append
a second ``Maxwell"  to the name of these equations to emphasize that Maxwell's vacuum law $D=E$, $B=H$
is being used to close the system of general Maxwell's equations for the electromagnetic fields.
This proves helpful if one needs to distinguish between different electromagnetic vacuum laws; 
(see e.g. \cite{Tah11}). }
system of equations for an electromagnetic spacetime $(\M,\bg,\bF)$; then we summarize the Ernst reduction
of these equations for the  special case that the spacetime possesses two commuting Killing fields; lastly we present the formulas
for the KN spacetimes and their electromagnetic fields.
Then in section 3,  we introduce our $G\to0$ limit of the KN electromagnetic spacetimes and raise the
question of the sources of their electromagnetic fields. In section 4 we present an ab-initio derivation of the solution of
the linear vacuum Maxwell equations on the topologically nontrivial $G\to0$ spacetime with the appropriate asymptotic
conditions at spatial infinity and at the ring singularity, which leads to the proper identification of the sources.
In section 5 we conclude our paper with a summary of our results and an outlook on the content of a companion paper \cite{KieTahIa}, where we  study the Dirac equation on the topologically nontrivial
zGKN spacetime.


\section{The KN family of electromagnetic spacetimes}
\subsection{Einstein--Maxwell equations}
We recall that an electromagnetic spacetime $(\cM,\bg,\bF)$ is a
solution to the Einstein--Maxwell system of equations:
\bna
R_{\mu\nu}[\bg] - {\ts\frac{1}{2}} R[\bg] g_{\mu\nu} & = & {\ts\frac{8\pi G}{c^4}} T_{\mu\nu}[\bg,\bF]\label{eq:Ein}\\
 \nab^\mu F_{\mu\nu} & = & 0.\label{eq:Max}
 \ena
 Here $R_{\mu\nu}$ denotes the Ricci curvature tensor and $R$ the scalar curvature of the metric $g$, $G$ is Newton's constant of universal gravitation, and $c$ is the {\em chronometric constant},\footnote{Commonly referred to as ``the speed of light in vacuum."  We find this terminology slightly problematic since the constant $c$ is present also in Einstein's {\em vacuum} equations, and without electromagnetism there is no ``light'' to speak of.} {relating time and space units}.   We will choose units in which $c=1$.
Finally, $T_{\mu\nu}$ is the energy(-density)-momentum(-density)-stress tensor of the electromagnetic field.  In the Maxwell--Maxwell's case, one has:
 \beq
 T_{\mu\nu} = {\ts\frac{1}{4\pi}}\left(F_{\mu}^\la {}*\!F_{\nu\la} -{\ts\frac{1}{4}} g_{\mu\nu} F_{\al\beta} F^{\al\beta}\right).
 \eeq
In particular, since the above energy tensor is trace-free, it is possible to omit the term proportional to scalar curvature $R$ in \refeq{eq:Ein} since $R$ will be zero.   Note that any solution $(\cM,\bg,\bF)$ to the above system will in particular depend on the value of $G$.



\subsection{Spacetimes with two commuting Killing fields and the Ernst reduction}
\subsubsection{Stationary solutions}
A solution triple $(\cM,\bg,\bF)$ is called {\em stationary} if there exists a Killing vectorfield $T$ for $(\cM,\bg)$ whose orbits are complete, diffeomorphic to $\RR$, and everywhere {\em timelike} in $\cM$, and such that its action leaves $\bF$ invariant, that is to say
$$ V := - \bg(T,T)>0,\qquad \cL_T \bg = 0,\qquad \cL_T \bF = 0.$$
 Since the electromagnetic field is assumed to be invariant under the action of $T$, one has
$$ d(i_T \bF) = - i_T d\bF = 0,\qquad d(i_T *\bF) = -i_T d*\bF = 0$$
where we have assumed that the component of Maxwell's equations \eqref{eq:Max} along $T$ is satisfied.  Let $\cU$ be a simply connected domain in $\cM$.  By Poincar\'e's lemma there are functions $\phi,\psi$ defined on $\cU$ such that
 $$ i_T \bF = d\phi,\qquad i_T *\bF = d\psi.$$
$\phi$ is called the {\em electric potential} and $\psi$  the {\em magnetic stream function}.
 Furthermore, let the one-forms $\be, \bb, \boldsymbol{\Ga}$ be defined as follows:
 $$ \be := i_T dT^\flat,\qquad \bb := i_T*dT^\flat,\qquad \boldsymbol{\Ga}:= \be+i\bb.$$
 In particular $\bb = *(T^\flat\wedge dT^\flat)$ is the {\em twist form} of $T$.  It is easy to see that $\be = dV$.  On the other hand, $\bb$ is not a closed 1-form, but it turns out that the contraction along $T$ of Einstein's equations \eqref{eq:Ein}, i.e. the equations $i_T \bR-\half R T^\flat = 2G i_T\bT$, impl{ies} that $$ d(\bb + \ka \boldsymbol{\xi}) = 0,$$
 where $$\ka :=2G$$  and $$ \boldsymbol{\xi} := \phi d\psi - \psi d\phi.$$
 Thus there is a function $Y$ on $\cU$ such that $\bb = dY - \ka \boldsymbol{\xi}.$
 It then follows that
 $$ \boldsymbol{\Ga} = dV + i dY - i \ka (\phi d\psi - \psi d\phi) = d\left(V + \frac{\ka}{2}(\phi^2+\psi^2 ) + i Y \right) - \ka (\phi -i\psi)d(\phi+i\psi).$$
 Therefore defining the following pair of complex-valued potentials (called {\em Ernst Potentials} after their discoverer F. J. Ernst \cite{Ernst}):
 $$ \cE := V + \frac{\ka}{2}|\Phi|^2 + iY,\qquad \cF := \phi+i\psi,$$
 one obtains that $ \boldsymbol{\Ga} = d\cE - \kappa \bar{\cF}d\cF.$
 One also notes that on $\cU$, $\mbox{Re} \cE - \frac{\ka}{2} |\cF|^2 = V \geq 0.$  Therefore the Ernst potentials take their values in a {\em Siegel domain} $\fS$ in $\Cset^2$:
 $$ \fS := \{ (\cE,\cF)\in \Cset^2\ |\ \mbox{Re}\cE \geq\frac{\ka}{2} |\Phi|^2 \}.$$

 Let $(\cN,\check{\bg})$ be the (3-dimensional,  {Riemannian}) quotient manifold of $(\cM,\bg)$ under the $\RR$-action  generated by $T$. $\cM$ can be viewed as a bundle over $\cN$, with projection $\pi:\cM\to \cN$ taking a point in $\cM$ to its orbit under the $\RR$ action, viewed as a point in $\cN$. It is always possible to introduce coordinates on $\cM$ that are adapted to the Killing field $T$, in such a way that the line element of $\bg$ reads
   \beq\label{lineg0} ds_\bg^2 = -V(dt +\boldsymbol{\alpha})^2 + \frac{1}{V} \ga_{ij} dx^i dx^j.\eeq  Here $t$ is a Killing parameter such that  $T = \frac{\p}{\p t},$ $(x^i)$ is an arbitrary coordinate system on $\cN$, $\boldsymbol{\al} = \frac{-1}{V} g_{0i} dx^i$ is a 1-form on $\cN$, $\check{g}_{ij} = g_{ij} +V\al_i\al_j$ is the metric induced by the projection on $\cN$, and $\boldsymbol{\ga}$ is a metric on $\cN$ which is conformal to $\check{\bg}$: $ \boldsymbol{\ga} := V \check{\bg}.$
\subsubsection{Stationary axisymmetric solutions}
Let us now assume that the original metric $\bg$ has another continuous symmetry, generated by another Killing field which is assumed to be {\em spacelike}  in $\cM$, with orbits of all points not fixed by its action being diffeomorphic to $\Sset^1$.  Let $\cG := \Rset \times SO(2)$ denote the full symmetry group.  The ``axis'' is now defined as the set of points $p$ in $\cM$ whose orbits under $\cG$ are {\em degenerate}, i.e. the isotropy group at $p$ is nontrivial.  We say that $(\cM,\bg,\bF)$ is {\em stationary and axially symmetric} if $\cG$ acts effectively on $(\cM,\bg)$ as an abelian group of isometries leaving $\bF$ invariant, and such that the orbits of points not on the axis are {\em timelike} 2-surfaces (cylinders).  We assume that the axis $\cA$ is nonempty.  Let $K$ be a generator of the abelian $\cG$ symmetry linearly independent from $T$. It can then be shown (e.g. \cite{Wei96}) that the distribution of the 2-planes which are the orthogonal complements of $\mbox{span}(K,T)$ in the tangent space at each point in the manifold is integrable, i.e. the two twist constants vanish. It then follows that the 2-dimensional quotient manifold $(\cQ,\tilde{\bg}) := (\cM,\bg)/\cG$  can be identified with a surface in $\cM$, which will have a nonempty boundary corresponding to the axis. We can once more take coordinates $t,\varphi$ on $\cM$ that are adapted to the Killing fields, i.e. $K = \p_\varphi$ and $T = \p_t$. The line element of $\bg$ then reads
\beq\label{lineg1} ds_\bg^2 = X d\varphi^2 + 2W d\varphi dt - V dt^2 + \tilde{g}_{ab}dy^a dy^b.\eeq
where $$X := \bg(K,K),\quad W := \bg(K,T),\quad V := - \bg(T,T).$$
The line element can also be written as follows
\beq\label{lineg2}
ds_\bg^2 = -V(dt +\mho d\varphi)^2 +\frac{\rho^2}{V} d\varphi^2 + \frac{1}{V} m_{ab} dy^a dy^b.
\eeq
Here
$$ \mho := \frac{W}{V},\quad \rho := \sqrt{W^2 +XV}, \quad m_{ab} :=V \tilde{g}_{ab}.$$
In particular $\rho d\varphi\wedge dt$ is the area element of the cylindrical group orbits (which are assumed to be timelike, hence $W^2+XV>0$).  Comparing \refeq{lineg2} with \refeq{lineg0} we obtain that $\boldsymbol{\alpha} = \mho d\varphi$, and also that
$$ ds_{\boldsymbol{\ga}}^2 = \rho^2 d\varphi^2 + ds_\calm^2,$$
thus $\calm$ is the metric induced on $\cQ$ from $\cN$.  Define the vectorfield
$$ \tilde{K} := K - \mho T$$
which is the orthogonal projection of $K$ onto $\cN$.  We have $\boldsymbol{\ga}(\tilde{K},\tilde{K}) = \rho^2$.  It can be shown (see \cite{Wei90}) that $\tilde{K}$ is a {\em hypersurface{-}orthogonal} Killing field for the metric $\boldsymbol{\ga}$, and thus $(\cQ,\calm)$ is a totally geodesic submanifold of $(\cN,\boldsymbol{\ga})$. Moreover,  $\rho$ as a function defined on $\cQ$ is {\em harmonic} (see \cite{Wei96} for a proof) i.e.
$$ \Del_\bm \rho = 0.$$
Let $z$ denote a conjugate harmonic function for $\rho$.  One can then coordinatize $\cQ$ using $(\rho,z)$.  These are called {\em Weyl coordinates}, they provide a nowadays so-called {\em isothermal}\,\footnote{First introduced by Lam\'e, who called them {\em thermometric parameters}.} system of coordinates for the quotient $\cQ$, i.e.
$$ ds_\bm^2 = e^{2u}(d\rho^2+dz^2)$$ where $u = - \log |d\rho|_\bm$.  Thus, in Weyl coordinates, the line element of a stationary axially symmetric electrovacuum metric can be put in the so-called {\em Lewis-Papapetrou} form
$$
ds_\bg^2 =  -V(dt +\mho d\varphi)^2 +\frac{1}{V}\left(\rho^2 d\varphi^2 + e^{2u}(d\rho^2+dz^2)\right).
$$
\subsubsection{Ernst equations}
Once the Ernst potentials $(\cE,\cF)$ are known, all the remaining unknown metric coefficients can be computed from them using quadratures (see \cite{Wei96} for details).  The potentials  themselves satisfy a system of PDEs known as the {\em Ernst equations} \cite{Ernst}:
\bea
\frac{1}{\rho}\nab\cdot(\rho d\cE) &=& \frac{1}{V}\boldsymbol{\Ga}\cdot d\cE\\
\frac{1}{\rho}\nab\cdot(\rho d\Phi) & = & \frac{1}{V}\boldsymbol{\Ga}\cdot d\Phi.
\eea
Here the covariant derivative and the inner product is with respect to the metric $\bm$ of the quotient manifold $\cQ$. Furthermore, the above can be viewed as equations for an axially symmetric {\em harmonic map} from $\RR^3$ into the Siegel domain $\fS$, which is a model for the complex hyperbolic space $\Hset_\Cset$ (see \cite{Maz83}.)

\subsubsection{Symmetry and gauge invariance}\label{symsec}
The target $\fS$ of the above mentioned harmonic map is a Hermitian symmetric space, $\fS \equiv SU(1,2)/S(U(1)\times U(2))$, on which the 8-dimensional Lie group $SU(1,2)$ acts as a group of continuous isometries consisting of, loosely speaking,  3 ``translations" $\tau_{z,\ga}$, one ``scaling" $\la_\beta$, one ``rotation" $\rho_\al$, and three  ``inverted translations" $\hat{\tau}_{z,\ga}$, with the following explicit actions \cite{Cha83}:
\bea
\tau_{z,\ga}(\cE,\Phi) & := & (\cE + 2z\Phi + |z|^2 + i \ga,\Phi+z),\qquad z\in \Cset,\ga\in \RR \\
\la_\beta (\cE,\Phi) & := & (e^{2\beta}\cE,e^\beta \Phi) ,\qquad \beta \in \RR\\
\rho_\al(\cE,\Phi) & := & (\cE, e^{i\al}\Phi),\qquad \alpha \in \RR \\
\hat{\tau}_{z,\ga}(\cE,\Phi) & := & \mathfrak{i} \circ \tau_{z,\ga} \circ \mathfrak{i} (\cE,\Phi),\qquad \mbox{where } \mathfrak{i} (\cE,\Phi) := (\frac{1}{\cE},\frac{\Phi}{\cE})
\eea
The rotation $\rho_\al$  in particular is what is usually referred to as a {\em duality rotation}.   Since the Ernst equations describe a harmonic map into a symmetric space, they are invariant under the above action of $SU(1,2)$, and thus solutions corresponding to $(\cE,\Phi)$ and $(\cE',\Phi') = g(\cE,\Phi)$ for any fixed $g\in SU(1,2)$ should be considered {\em physically equivalent}. In addition to these continuous isometries, there is also a purely discrete one, namely complex conjugation  $\mathfrak{c}(\cE,\Phi) = (\bar{\cE},\bar{\Phi})$, with a similar requirement for physical equivalence.  Note that $\mathfrak{c}^2 = id$.  Now, from the above description of the action of the group $SU(1,2)$ on the Ernst pair of potentials $(\cE,\Phi)$ it is clear that there is only one non-identity element $g\in SU(1,2)$ with an {\em involutive} action, i.e. with $g^2 = id$, namely $g=\rho_\pi$. 
 
 Consider an involutive isometry $\tau$ of the extended manifold, e.g. a reflection $\tau:\tilde{\cM} \to \tilde{\cM}$, such that $\tau^2 =id$.
For a potential $\Phi$ originally defined on $\cM$ to have an extension to $\tilde{\cM}$ it would then be necessary that $\Phi \circ \tau$ is equivalent to $\Phi$, i.e. there should exist a $ g \in \mbox{span}(SU(1,2) \cup \{\mathfrak{c}\})$ such that $\Phi \circ \tau = g \Phi$.  But $\tau^2 = id$ and thus $g^2 = id$, so that $g \in \Zset_2 \times \Zset_2$, generated by $\{\mathfrak{c},\rho_\pi\}$.  It follows that under such involutive isometries, $\Phi$ can at most change  by a sign, or by conjugation, or both.

\subsubsection{Brief history of multisheeted spacetimes}
A stationary spacetime is {\em static} if the timelike Killing field $T$ is hypersurface-orthogonal, i.e. $\bb \equiv 0$.  Static, axisymmetric solutions of the Einstein Vacuum or Einstein--Maxwell equations are known as {\em Weyl Solutions} \cite{Wey17}.  In the vacuum case, by setting $V = e^{2v}$ in the Ernst equation one finds that $v$ is an axisymmetric {\em harmonic function} on $\RR^3$.  If the spacetime is asymptotically flat, $v$ must vanish at infinity, and therefore needs to blow up somewhere in the interior.  The blow-up set of $v$ may either be a horizon or an actual singularity.  Since $v$ is axisymmetric, the simplest possibilities for the blow-up set are a point or an interval on the axis of symmetry, or a circle in a plane orthogonal to the axis.  Well-known members of the Weyl class of solutions are the Curzon solution (singular set is one point), the Schwarzschild (singular set is an interval, corresponding to the event horizon) and the Bach-Weyl solution \cite{BacWey22}, which has a ring singularity (but no exotic topology).  In 1964 David Zipoy found another family of Weyl solutions by solving Laplace's equation in spheroidal coordinates, and noticed that his solution can be extended to a double-sheeted spacetime by allowing the spheroidal coordinate $r$ to take on negative values, thus discovering the first multi-sheeted solution of Einstein's equations.  The Kerr solution, which is not a member of the Weyl class, and which had just been discovered in 1963,  was not known to have such properties until its maximal analytical extension was given by Boyer and Lindquist in 1967, using spheroidal coordinates.  Meanwhile the Kerr--Newman solution was discovered in 1965, and the multi-sheetedness of its maximal analytical extension was revealed by Carter in 1968.  
\subsection{The KN metric}
The {\em KN family of spacetimes} is a three-parameter family of stationary axially symmetric solutions to (\ref{eq:Ein}-\ref{eq:Max}).  The 3 parameters are commonly called 
``charge" $\sq$,  {ADM} ``mass" $\sm$ and  {ADM} ``angular momentum per unit mass" $a$ of the spacetime\footnote{Even though these terms classically are associated with matter, and these spacetimes are vacuum, i.e devoid of matter. ADM stands for Arnowitt, Deser, and Misner \cite{ADM61} who formulated these quantities as surface integrals on the sphere at infinity of a spatial slice}.  The causal structure of these spacetimes depends crucially on whether the quantity
$$\scp^2 := a^2+\ka\sq^2-\ka^2\sm^2$$
is negative (called the subextremal case), zero (extremal), or positive (hyperextremal).  Since our goal is to study the limiting case $\ka\to 0$, we are thus confined to the hyperextremal case, where $\scp$ is real, and without loss of generality positive.

Let $\cQ$ be a half-plane, with Cartesian coordinates $(\rho,z)$, $\rho \geq 0$.  Let $(r,\theta)$ be the following system of implicitly defined elliptical coordinates on $\cQ$:
    \beq\label{coords}    \rho = \sqrt{(r-\ka\sm)^2 +\scp^2}\ \sin\theta,\qquad z = (r-\ka\sm)\cos\theta,
    \eeq
  Thus the level sets of $r$ are confocal ellipses that shrink to the line segment $[-\scp,\scp]$ on the $\rho$ axis, and the level sets of $\theta$ are the  family of hyperbolas orthogonal to those ellipses. The coordinates $r$ and $\theta$, together with the azimuthal angle $\varphi$ (the Killing parameter for $K$) form a system of {\em oblate} spheroidal coordinates\footnote{A similar construction in the {\em subextremal} case gives rise to a system of {\em prolate} spheroidal coordinates covering the exterior patch.} for the manifold $\cN$.  There is a coordinate singularity at $r = \ka \sm$ corresponding to the ellipsoids degenerating into a line segment.
    The axis of rotation is where $\sin\theta = 0$.
    Let us define the auxiliary quantities
    $$ \Delta := r^2 - 2\ka\sm r +\ka \sq^2 + a^2 = (r-\ka\sm)^2 + \scp^2 >0,\quad \Sigma := r^2 + a^2 \cos^2 \theta \geq 0.$$
    For the KN family, the metric of the quotient space is
$$ds_\bm^2 = \frac{V\Sigma}{\Delta}( dr^2 + \Delta d\theta^2),$$
 and the coefficients of the metric of the group orbits are
 \beq\label{XVW}
    X = \frac{(r^2 + a^2)^2 - \Delta a^2 \sin^2\theta}{\Sigma}\sin^2\theta,\quad V = \frac{\Delta-a^2\sin^2\theta}{\Sigma},\quad W = \frac{a\sin^2\theta(2\ka\sm r - \ka\sq^2)}{\Sigma}
    \eeq
so that the area element of the group orbits comes out to be
$$\sqrt{XV+W^2}d\varphi\wedge dt = \sqrt{\Delta}\sin\theta d\varphi\wedge dt = \rho d\varphi\wedge dt$$
just as it should.    The Ernst potentials  for KN are
\bna\label{KN}
\cF &=& \frac{\sq}{r - i a \cos\theta} =\frac{\sq  r}{\Sigma} + i  \frac{\sq a\cos\theta}{\Sigma} \\
\cE &=& V - \frac{\ka}{2} |\cF|^2 + iY = 1 - \frac{2\ka\sm}{r - i a \cos\theta} = 1 - \frac{2\ka\sm r}{\Sigma} - i \frac{2\ka\sm a\cos\theta}{\Sigma}.
\ena
Note that the image of $\cQ$ under the Ernst map $(r,\theta) \mapsto (\cE,\cF)$ is contained in a complex line in the Siegel domain: $\cE + \frac{2\ka\sm}{\sq} \cF = 1$.   This is but one indication among many that the KN solution is algebraically special.

All KN spacetimes have a discrete symmetry, namely the metric is invariant under reflection with respect to the equatorial plane $\theta=\pi/2$.  In other words $(t,r,\theta,\varphi)\mapsto (t,r,\pi-\theta,\varphi)$ leaves the metric unchanged.  Note that the electric potential is symmetric while the magnetic stream function is anti-symmetric with respect to this reflection.

The great advantage of using oblate spheroidal coordinates\footnote{In general relativity context often referred to as {\em Boyer--Lindquist} coordinates in recognition of the researchers who used them to find the maximal analytical extension of the Kerr solution \cite{BoyLin67}.}, apart from the obvious simplicity of the formulas for the metric coefficients and electromagnetic potentials, is that they lead to an  {\em extension} of the  KN manifold, upon noticing that the coordinate $r$ in the above is allowed to be negative\footnote{This property of spheroidal coordinates seems to have first been observed and remarked on by Zipoy \cite{Zip66}.}.  In fact $(r,\theta) \mapsto (2\ka\sm - r,\pi -\theta)$ leaves $(\rho,z)$ unchanged\footnote{For $\ka\sm\ne 0$, this is {\em not} a discrete symmetry of the spacetime, since there is a singularity at $r=0$ but none at $r=2\ka\sm$.}.  Thus every point $(\rho,z)$ in the quotient $\cQ$, with the exception of $(\scp,0)$, corresponds to {\em two} points in an extended version of the quotient, namely $(r,\theta)$ and $(2\ka\sm - r,\pi-\theta)$.  In the following we will show that these spheroidal coordinates are in fact coordinates on the {\em maximal extension} of the spacetime manifold.

\subsection{The topology of KN spacetimes}
As it's clear from \refeq{XVW}, in coordinates in which the metric is given, there is a singularity at $\Sigma =0$, i.e. at $r=0, \theta = \pi/2$, for all $t$ and $\varphi$.  Let $$\cR := \{(t,r,\theta,\varphi)\ | \  \Sigma(r,\theta) = 0 \}$$ and let $\cR_t$ denote the $t$-slice of this surface. $\cR$ is diffeomorphic to a cylinder, and $\cR_t$ to a circle (a {\em ring.})  It is not hard to check, for example by calculating curvature invariants of the metric, that the surface $\cR$ is a true singularity for the metric, not just a coordinate one.  It was shown by Carter \cite{Car68} that all geodesics in this spacetime are complete, except for those that run into $\cR$.  He also showed that, in the hyperextremal case (which is the case under study in this paper,) the single coordinate patch $(t,r,\theta,\varphi) \in \RR\times \RR \times [0,\pi) \times [0,2\pi)$ already covers the maximal analytical extension $\tilde{\cM}$ of the spacetime manifold, which according to Carter is  simply connected as a topological manifold.
     
The presence of the singular surface $\cR$  however means that in order to consider the spacetime as a {\em Lorentzian} manifold, this singularity needs to be excised, i.e. let $\cM^* := \tilde{\cM}\setminus \cR$.  The resulting manifold $\cM^*$ is then certainly non-simply connected:  at any fixed $t$ there are loops that are not homotopic to zero since they ``thread" the excised ring $\cR_t$.

Construction of a two-sheeted extension for $\cM^*$ is carried out in \cite{HawEll73}, but the main idea is already present in \cite{Zip66}:
Such a 2-sheeted extension is first constructed for $\cN^*$.  At the singularity, we have $\rho = \rho_0 := \sqrt{a^2+\ka \sq^2}$ and $z = 0$.  $\rho_0$ is often referred to as the (Euclidean) ``radius" of the ring singularity.  Note that $\rho_0 \to |a|$ as $\ka \to 0$.  The extension for $\cN^*$, which will be denoted by $\fN$, consists of two copies of $\RR^3$ with Cartesian coordinates $(x_1,x_2,x_3)$ and $(x'_1,x'_2,x'_3)$; and the map $\Pi:\fN \to \cN^*$ identifies the disks $\cD = \{ x_1^2 + x_2^2 < \rho_0^2, x_3 = 0\}$ with the corresponding disk $\cD'$ in the second copy in such a way that the top of $\cD$ is identified with the bottom of $\cD'$, and vice-versa.  If we endow each copy of $\RR^3$ with oblate spheroidal coordinates in such a way that $r\geq \ka \sm$ corresponds to one copy and $r\leq \ka \sm$ to the other, then the map $\Pi$ is given simply by \eqref{coords}, i.e. $\Pi(r,\theta,\varphi) = (\rho(r,\theta),z(r,\theta),\varphi)$. The {\em involution} $\tau$ is then defined to be $\tau(r,\theta,\varphi) = (2\ka\sm - r, \pi-\theta,\varphi)$.  Clearly $\tau^2 = id$. Taking the Cartesian product of $\fN$ with $\RR$ will then provide an extension for $\cM^*$, which turns out to be the maximal analytical extension of $\cM^*$.  
We will denote it by $\fM$.

\subsection{Sommerfeld spaces}
The manifold $\fN$, which can be identified with a constant $t$ slice of the spacetime $\fM$ defined above, is the simplest nontrivial example of a {\em Sommerfeld space}.  

A Sommerfeld space by definition consists of $m$ copies of $\RR^3$ (called {\em sheets} or {\em leaves}) that are glued together along $n$ two-sided surfaces differomorphic to the 2-disk, called {\em branch cut surfaces}, that span $k$ distinct non-linked simple closed curves, called {\em branch curves}.  See \cite{Eva51} for the precise definition.  Such a space is implicitly present in Sommerfeld's seminal paper \cite{Som97} on multi-valued harmonic functions (called ``branched potentials" by Sommerfeld).  His main idea was to generalize the concept of a Riemann surface to higher-dimensional manifolds, which Sommerfeld termed {\em branched Riemann spaces}, by looking for a three-dimensional space on which a given multi-valued harmonic function becomes single-valued.  Sommerfeld's interest in these spaces grew out of his study of the diffraction problem \cite{Som96}.   The concept of a branched space was subsequently generalized by G. C. Evans \cite{Eva51}.  

Apparently unaware of these developments, Zipoy was led to the same concept by noticing the behavior of spheroidal coordinate functions used to describe his newly-discovered static axisymmetric family of solutions to Einstein's vacuum equations \cite{Zip66}.  Zipoy appears to have been  the first person in the GR community to notice the ``duodromic" nature of space in the vicinity of a branch curve, i.e. that a circular path around a point on the curve would have to take two complete turns before it can close {(for which reason we will speak of a spacetime with this property  as having a {\em Zipoy topology})}.  
\subsection{The causal geometry of KN spacetimes}
The causal structure of KN spacetimes in all three regimes of the parameters was thoroughly investigated by Carter \cite{Car68}.  The following conclusions are already evident from \refeq{XVW}:  In the hyperextreme case the Killing field $T$ remains timelike throughout the manifold, i.e. $V>0$.  On the other hand, the rotation Killing field $K$ undergoes a change of type.  It is spacelike for large positive and negative $r$, while close to the ring singularity, i.e. for $r$ small, it becomes {\em timelike}.  Let
$$\cT = \{(r,\theta,\varphi)\ |\ K\ne 0, \bg(K,K)= 0\}\subset \cN^* $$
denote the surface on which $K$ is null.  $\cT$ is diffeomorphic to a 2-torus  whose {\em soul} is the ring singularity $\cR$.  Inside $\cT$, the Killing field $K$ is timelike, hence its orbits are closed timelike curves, which violates causality. On the other hand, we have that the determinant of the metric induced on the group orbits $-VX-W^2 =  -\rho^2 <0$ and therefore the group orbits are {\em always} timelike cylinders, which is another way of saying that there are no horizons in this spacetime.   Carter showed \cite{Car68} that  the existence of closed timelike loops, combined with the absence of horizons, imply that  the entire spacetime manifold is a causally vicious set, i.e. any two points in it can be connected by both a future and a past-directed timelike curve!
 For this reason we call $\cT$ the {\em causality limit surface}.
\section{The $G\to 0$ limit of KN spacetimes}
For a fixed value of the coupling constant $\ka = 2 G$, the KN solution $(\cM^*,\bg,\bF)$ is a 3-parameter family, indexed by mass $\sm$, charge $\sq$ and angular momentum per unit mass $a$.  Fixing these parameters (as well as  $c$, which we have fixed to be one), and taking the formal limit of the solution {represented in Boyer--Lindquist coordinates}
 as $G \to 0$ (or equivalently $\ka \to 0$), one arrives as the triple $(\cM_0^*,\bg_0,\bF_0)$.  Below we study the properties of this limiting solution.
\subsection{Topology and geometry of the limiting solution}
It is clear that in {this} limit $\ka\to 0$, the topology of $\cM^*$ remains nontrivial, since the singular surface $\cR$ does not disappear in this limit (i.e. the ``radius of the ring" $\rho_0 \to |a| \ne 0$ as $\ka \to 0$).  Thus $\zM^*$ is still not simply connected.  The limiting
$\fM_0$  provides an extension for $\zM^*$.  Let $\cN_0^*$ denote the $t=0$ slice of $\cM_0^*$.   In the next subsection we will give an alternative construction of an extension for $\cN_0^*$, using {\em only one copy} of $\RR^3$, which will help in visualizing the limiting electromagnetic field $\bF_0$.

But first we observe that  $\zM^*$ is flat:  Taking the limit $\ka \to 0$ of the KN metric  {represented in Boyer--Lindquist coordinates}, we obtain
$$
ds_{\zg}^2 = -dt^2 + (r^2+a^2)\sin^2\theta d\varphi^2 + \frac{r^2+a^2\cos^2\theta}{r^2+a^2} \left(dr^2+ (r^2+a^2) d\theta^2\right).
$$
Recalling how the Weyl coordinates $(\rho,z)$ are related to $(r,\theta)$:
$$\rho = \sqrt{r^2+a^2}\sin\theta,\qquad z = r \cos \theta,$$
we may thus compute that
$$
ds_{\zg}^2 = -dt^2 + dz^2 + d\rho^2 + \rho^2 d\varphi^2,
$$
which is the Minkowski metric.  

Note that since the above coordinate transformation is singular at the ring $\cR_0 = \{(r,\theta,\varphi)\ |\ r=0,\theta=\frac{\pi}{2}\}$, this calculation does not tell us anything about whether the ring is still a singularity of the limiting metric $\zg$.  The fact that it still is can be established in the following way.  For $\ep>0$ small  let $\ga:[0,2\pi] \to \fM_0$ denote the small ``circle" lying in a meridional plane (which, without loss of generality we can take to be the plane $\varphi = 0$)  and centered at a point $(0,\frac{\pi}{2},0)$ on the ring:
$$
\ga(\al) = (r(\al),\theta(\al),0)\mbox{ with }\ r=\ep a \cos\al,\ \theta=\ep \sin\al + \frac{\pi}{2}.
$$
 The circumference of this circle is
$$ \int_0^{2\pi} \sqrt{\bg_0(\dot{\ga},\dot{\ga})} d\al = 2\pi a\ep^2 + O(\ep^3)$$
while the distance of a point on the circle from its center is
$$ \sqrt{(\rho-\rho_0)^2 + z^2} = \frac{a}{2}\ep^2 + O(\ep^3).$$
Thus the limit of the ratio of circumference to radius of the circle $\ga$ as it shrinks to a point is equal to $4\pi$, indicating that every point on the ring is a {\em conical singularity} for the metric $\zg$.

We also note that as $\ka\to 0$ the causality limit surface $\cT$ shrinks down to its soul and disappears in the limit, since $X\to (r^2+a^2)\sin^2\theta $ as $\ka \to 0$ and in particular $X\to a^2>0$ on the ring. Thus, the singular ring is {\em spacelike}.  There are no other singularities or pathologies remaining in $\fM_0$.

The  factor $4\pi$ obtained above also establishes that $\fM_0$ is still ``double-sheeted" just like $\fM$ was. The involution is now simply $$\tau_0(r,\theta,\varphi) = (-r,\pi-\theta,\varphi).$$
 Note that the metric is now invariant under this transformation, so that this is a {\em new} discrete symmetry of the spacetime.  The electromagnetic KN potentials are {\em anti-symmetric} with respect to $\tau_0$.  This shows that in the study of the KN electromagnetic fields, one cannot confine one's attention just to one of the two sheets of the extended manifold $\fM_0$. It is necessary to view the fields as defined on the whole extended manifold.

\subsection{An alternative system of coordinates for $\fM_0$}
An extension $\fN'_0$ for $\cN_0^*$, the $t=0$ slice of $\cM_0^*$, can be constructed using only one copy of $\RR^3$  in the following way:  For $a \in \RR$ consider the Riemann surface
$$ \cS := \{ (Z,W)\in \Cset^2\ |\ W^2 - 2WZ + a^2 = 0\}.$$
This is a non-singular, two-sheeted surface with branch points at $Z = \pm a$, consisting of two copies of $\Cset$ which are slit  at $\{ W=u+iv \in \Cset\ |\ -a\leq u \leq a,\ v = 0\}$ and glued together in the usual manner, with the top lip of the slit in one copy glued to the bottom lip in the other copy.  Since $\cS$ is a {\em parabolic} Riemann surface, its universal cover is still $\Cset$, and its compactified cover is the Riemann sphere $\bar{\Cset}$.  The  projection $\pi:\Cset \to\cS$ is given by the 2-to-1 map $W \mapsto Z = \pi(W)=\frac{W^2+a^2}{2W}$.   Let $Z=\rho+iz$ and let $(r,\theta)$ denote the following elliptical coordinates in the $Z$ plane:
$$ \left\{\begin{array}{rcl} \rho & = & \sqrt{r^2+a^2} \sin\theta \\ z & = & r \cos\theta\end{array}\right.$$
Also let $(\tilde{r},\tilde{\theta})$  denote {\em  polar coordinates} in the $W$ plane, i.e. $W=\tilde{r}e^{i\tilde{\theta}}$.  From the expression for $\pi$ it thus follows that
$\theta = \tilde{\theta}$ and \beq\label{req} r =\frac{\tilde{r}^2 - a^2}{2\tilde{r}}.\eeq

The mapping $\pi$ is extended to $\RR^3$ by making it {\em equivariant} with respect to rotations about the $z$-axis, i.e., letting
$(\theta,\varphi)$ denote the usual {\em spherical} coordinates on $\Sset^2$, with $\theta$ the polar and $\varphi$ the azimuthal angle,  the mapping is simply
$$\bar{\pi}(\tilde{r},\theta,\varphi) = (\frac{\tilde{r}^2+a^2}{2\tilde{r}}\sin\theta,\frac{\tilde{r}^2-a^2}{2\tilde{r}}\cos\theta,\varphi).$$
 It is easy to check that the range of $\bar{\pi}$ is $\cN_0^*$, the $\ka\to 0$ limit of the 3-dimensional manifold $\cN^*$ in the previous section, formed by gluing two copies of $\RR^3$ along a 2-disk. We denote the domain of $\bar{\pi}$ by $\fN'_0$. It is homeomorphic to (a single copy of) $\RR^3$. Finally, the map $\bar{\pi}$ is extended with respect to $t$ in the obvious way, to arrive at $\hat{\pi}(t,\tilde{r},\theta,\varphi) = (t,r,\theta,\varphi)$.  The target of $\hat{\pi}$ is now the limiting spacetime $\zM^*$.  The domain of $\hat{\pi}$ is the extended spacetime $\zfM'$ we were looking for. Topologically, it is $\RR^4$.

We can check that under $\hat{\pi}$, the top and bottom hemispheres of the sphere $\tilde{r} = a$ are mapped to the two disks $\cD$ and $\cD'$ in the extended manifold $\cN_0^*$, the equator of the sphere is mapped to the ring $\cR_0$, while the interior and the exterior of this sphere each get mapped to one of the sheets in $\N_0^*$.  The involutive transformation between the sheets is now given by the {\em inversion} with respect to the sphere of radius $a$:
$$
\tau_1(\tilde{r},\theta,\varphi) = (\frac{a^2}{\tilde{r}},\pi-\theta,\varphi).$$
The advantage of this realization of the extension is that the ``gluing" is automatic, meaning that the two sheets are glued exactly where they come in contact with each other, i.e. on the surface of the sphere. Note  that the  ring $r=0,\theta=\frac{\pi}{2}$ is fixed by $\tau_1$.

The pullback under $\hat{\pi}$ of the metric of $\zM$ is denoted by $\hat{\bg} := \hat{\pi}^*\bg_0$, and we compute its line element to be
$$
ds_{\hat{\bg}}^2 = -dt^2 + \frac{(\tilde{r}^2-a^2)^2+4a^2\tilde{r}^2\cos^2\theta}{4\tilde{r}^4} (d\tilde{r}^2 + \tilde{r}^2 d\theta^2) +\frac{(\tilde{r}^2+a^2)^2 \sin^2\theta}{4\tilde{r}^2} d\varphi^2.
$$
$\hat{\bg}$, being the pullback of a flat metric, is flat. Note however that its line element does not coincide with that of the Minkowski metric in spherical coordinates, that is to say, $(\tilde{r},\theta,\varphi)$ are {\em not} standard spherical coordinates (in particular $\tilde{r}/r \sim 2$ for large $r$).  
\subsection{The KN electromagnetic fields and the ``problem of sources"}
It is a remarkable feature of the KN solution (already noticed in \cite{Tio73}) that the electromagnetic potentials $(\phi,\psi)$
for this solution do not depend explicitly on the coupling constant $\ka$.  They read\footnote{Incidentally, these were originally discovered by P. Appell \cite{Appell}
as real and imaginary parts of his complexification of the electrostatic
Coulomb potential. Appell did not notice, though, that they live on a
double-sheeted space.}:
$$\phi = \frac{\sq  r}{\Sigma},\qquad \psi = \frac{\sq a\cos\theta}{\Sigma}.$$
 It thus follows that, as $G\to 0$ these same expressions continue to be solutions of Maxwell's equations, on the background $\zM$, which as we have shown is flat, but with a nontrivial topology.  This observation (that the fields are independent of $G$) may be used to ``derive" the KN fields from classical flat-space electromagnetism: one first guesses what the right sources should be, then solves the linear Maxwell equations with those sources to find the corresponding stationary solution, which one then checks is identical to the KN fields in some coordinates.  This is the approach in \cite{Tio73} for example. The question to ask therefore, is what kind of ``sources", i.e. what arrangement of charge and current distributions, can possibly give rise to the electromagnetic field associated with the KN potentials?  Since these potentials are singular {\em only} where $\Si = 0$, i.e. on the ``ring" $\cR_0$, it stands to reason that one looks for charge/current distributions concentrated only on this ring.

Previous attempts at solving the so called ``problem of sources for KN" however,
suffer from what appears to be a certain lack of enthusiasm for non-trivial topology among the researchers working in this area\footnote{Similar attitudes in the larger GR community may have been responsible for Zipoy not receiving enough credit for his discovery.  The following comment \cite{BonSac68} is perhaps representative of this attitude: ``D. M. Zipoy has presented some static axially symmetric solutions in spheroidal coordinates of Einstein equations for empty space $R_{ik} = 0$.  He endowed them with rather terrifying topological properties."}, as a result of which, the KN fields were invariably forced to reside in Minkowski spacetime, hence causing them to have singularities across an entire disk spanning the ring, (in much the same way that the argument function in one-complex variable theory would appear to have a jump along a ray in the complex plane.)  This neglects the fact that the electromagnetic fields in question are defined on the {\em extended spacetime}, which is double-sheeted, and that indeed they are solutions to Maxwell's equations, but with stationary sources prescribed on a time slice in that manifold.   As a result of this oversight, the sources proposed, e.g. by W.~Israel in \cite{Isr70}, for the KN fields were supported not just on the ring, where the fields are obviously singular, but on the whole disk that spans this ring, where no such singularity is apparent in the fields.  This resulted in extremely exotic and self-contradictory properties for the proposed sources, which were then taken as starting point for later researchers in the field\footnote{In \cite{PekFra87} for example, we find the following recipe proposed for the ``KN source":
\begin{quotation}\small
...a system of currents and of electric surface charges which are distributed over a circular disc of radius $a$, centered at the origin, and oriented normally to the direction of the angular momentum vector.  For a positive total electric charge $e$, the surface charge density is negative inside the disc, becoming infinitely negative as the rim of the disc is approached.  On the rim, there is a positive charge of infinite density which not only neutralizes the negative charge distributed in the interior of the disc, but also leaves a residue of a positive charge equal to $e$.  Similarly, the currents flow in the negative direction inside the disc, with a current that becomes infinitely negative at the rim.  On the rim there is a positive current of infinite intensity, which generates a magnetic moment compensating the negative magnetic moments distributed in the interior, and leaves a net integrated magnetic moment of magnitude $ea$, the latter being equal to the dipole component of the total magnetic moment...
\end{quotation}
It is important to note that, despite its implausibility,  what is being described is precisely the electromagnetic properties of the KN source as proposed first in \cite{Isr70}.  What is not being described in the above paragraph is that, in order for this to be a source, not just for the KN  electromagnetic field but also for the spacetime {\em metric} as well, the above disc has to, in addition, carry a mass  of equally bizarre proportions.}.

Apart from the implausibility of acquiring material for Israel's source ``in the shops", another obvious problem with this recipe is that there is no geometric way to distinguish the ``disc" on which these distributions are to be prescribed; in fact {\em any} 2-surface that spans the ring $\cR_0$ should work equally well.  This is clearly a very good indication that  one has started from the wrong premises.  In the next section we will demonstrate that, at least for the limiting member $(\zM,\zg,\zF)$ of the KN family, it is possible to exhibit a source which is {\em only} supported on the singular ring, a source that lives in {\em both sheets} of the extended manifold $\fM_0$, not just in one, and which easily gives rise to the correct field $\zF$, with no added singularities and/or causal pathologies.  


\section{\emph{Ab initio} derivation of the limiting KN electric and magnetic fields}
\subsection{The equation}
Consider Maxwell's equations (in  the absence of {\em regular} sources) on $\zfM$:  These are
$$ d\bF = 0,\qquad d*\bF = 0$$
where $\bF= d\bA$ is the electromagnetic Faraday tensor, and the equations are understood to hold in the sense of distributions (thereby allowing for {\em singular} sources to be there).   Let $T$ denote a timelike Killing field for $\zfM$ with $\RR$-orbits.  A {\em stationary} field $\bF$ satisfies $\cL_T \bF = 0$.  This gives rise to the existence of two potentials $\phi$ and $\psi$ on $\zfM$ such that
$$ i_T\bF = d\phi,\qquad i_T*\bF = d\psi.$$
Both $\phi$ and $\psi$ satisfy Laplace's equation with respect to the metric $\hat{\bg}$ of $\zfM$:
$$ \Del_{\hat{\bg}} \phi = \Del_{\hat{\bg}}\psi = 0.$$
Suppose now that $K$ is a {\em spacelike} Killing field for $\zfM$ with $\Sset^1$-orbits, that $\bF$ is also invariant under the action of $K$, i.e. $\cL_K \bF = 0$, and that $[T,K] = 0$. Let $t$ and $\varphi$ denote coordinates adapted to these two Killing fields, i.e. $K = \frac{\p}{\p\varphi}$ and $T = \frac{\p}{\p t}$.  The potentials $\phi,\psi$ are therefore independent of $t$ and of $\varphi$.  By the discussion of the previous section, each potential is  a critical point of the Dirichlet energy functional $\hat{E}$:
$$\hat{E}(\hat{f}) = \half \int |\nab \hat{f}|^2 e^{2\la} d^3x
= \pi \int \left( |\hat{f}_{\tilde{r}}|^2 + \frac{1}{\tilde{r}^2} |\hat{f}_\theta|^2 \right) \frac{\tilde{r}^2+a^2}{2\tilde{r}^2}\ \tilde{r}^2 \sin \theta d\tilde{r} d\theta.
$$
Let $$w := e^{\la} \hat{f}.$$  We then have
$$ \hat{E}(\hat{f}) = \int |\nab w|^2 + p(x) w^2 \ d^3 x,$$
where $p(x)$ is the following (smooth, positive, and radial) potential
$$ p(x) := \Del_{\hat{\bg}} \la + |\nab \la|_{\hat{\bg}}^2 = \frac{a^2}{(\tilde{r}^2 + a^2)^2}.$$
Conclusion:  both $\phi$ and $\psi$ are of the form $e^{-\la} w$ where $w$ is an axially symmetric ($\varphi$-independent) solution of the following equation
$$ -\Del w + p(x) w = 0$$
with $\Del$ the standard Laplacian in $\RR^3$.  Below we shall see that the potential term can be transformed away by going into another coordinate system.

\subsection{The boundary and asymptotic conditions}
Recall that  $\Phi = \phi+i\psi$ is defined on the extension $\zfM$ of the spacetime $\cM_0$.  By the discussion in \ref{symsec} $\Phi\circ \tau_1$ can only be equal to one of the following: $\Phi$, $-\Phi$, $\bar{\Phi}$, and $-\bar{\Phi}$, there is no other possibility.
Suppose we require $\Phi$ to be {\em anti-symmetric}, i.e.
\beq\label{sym1}\phi \circ \tau_1 = - \phi,\qquad \psi \circ \tau_1 = -\psi.\eeq
We recall moreover, that all KN spacetimes, including the zero-$G$ limit $\M_0$ have a discrete reflectional symmetry with respect to the equatorial plane $\theta = \pi/2$.  We further require the electric potential $\phi$ to be symmetric, and the magnetic stream function $\psi$ to be anti-symmetric, with respect to this reflection, i.e.
\beq\label{sym2} \phi(\tilde{r},\pi -\theta) = \phi(\tilde{r},\theta),\qquad \psi(\tilde{r},\pi-\theta) = -\psi(\tilde{r},\theta).\eeq
It follows from \refeq{sym1} and \refeq{sym2} that $\phi$ (if finite) must vanish on the surface of the sphere $\tilde{r}=a$, while $\psi$ (if finite) should vanish on the plane $\theta=\pi/2$:
$$ \phi(a,\theta) = 0,\qquad \psi(\tilde{r},\frac{\pi}{2}) = 0.$$
 In addition we postulate that these potentials vanish at infinity as well, i.e.
$$\lim_{\tilde{r}\to \infty} \phi(\tilde{r},\theta) = 0,\qquad \lim_{\tilde{r}\to \infty} \psi(\tilde{r},\theta) = 0.
$$
Since the potential $p$ is positive, it follows from the maximum principle that one must prescribe a singularity for $\phi$ and $\psi$, otherwise the only solution would be that they are identically zero.  Since $\phi$ and $\psi$ are axially symmetric, the minimal singular support would be  a point on the axis, and if there is an off-axis singularity, then there must be  a circle worth of them.

Next we observe that the singular behavior of the KN electric potential $\phi$ in the $(\rho,z)$ variables is
$$ \phi \sim \frac{c}{\sqrt{|\rho - a|}}\mbox{ as } \rho \to a,$$
which is neither consistent  with a uniform distribution of charged {\em monopoles} on a  ring in a flat spacetime (that behavior would be $\phi \sim c\log |\rho -a|$)  nor is it consistent with a  distribution of {\em dipoles} on the ring (i.e. $\phi \sim c/|\rho-a|$) but in fact it is something in between, in other words, it is a {\em sesqui}-polar singularity.  Indeed, a plotting of the KN electric lines of force (i.e. the direction field for $-\nab \phi$) shows  its asymptotic monopole behavior near infinity (in each sheet), but also a marked departure from either monopole or dipole behavior near the ring, and the presence of {\em four} saddle points on the symmetry axis.  A similar analysis can be performed for the KN magnetic field, revealing an asymptotic dipole in each sheet but  a central singularity inconsistent with any multipolar behavior.  Figure~1 shows the trace in the $(\tilde{r},\theta)$ plane of the KN electric and magnetic fields on a meridional plane $\varphi=$const.  The circle in the two figures is the trace of the sphere $\tilde{r}=1$, where the two sheets are glued together, one sheet being inside and the other one outside that sphere. 

\begin{center}
\begin{figure}[h]
\hspace{.375in}
\includegraphics[scale=.5]{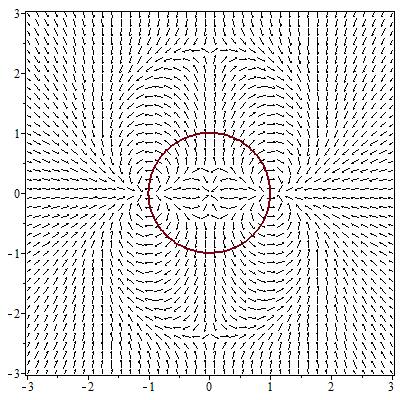}
\includegraphics[scale=.5]{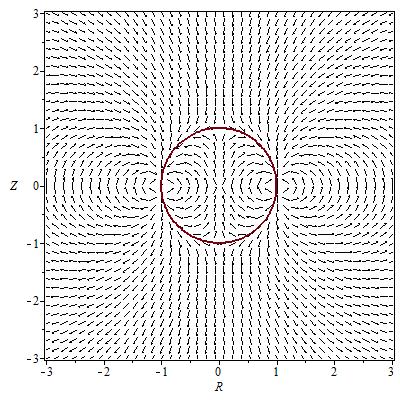}
\caption{\label{fig:zam}Unit tangent fields for the lines of force of the KN electric (left) and magnetic fields (right)}
\end{figure}
\end{center}
Therefore, instead of  characterizing the singular sources in terms of their pole rank, i.e. in terms of Dirac $\delta$-distributions supported on the ring or their distributional derivatives, in the next subsection we will provide a mathematical formulation of a Dirichlet problem with {\em limiting} boundary values for the potentials, which, we will then show, uniquely determines the solution.

\subsection{Ring-centered coordinates, complex formulation, and uniqueness}
It is possible to use complex variable theory to study the above mentioned Dirichlet problems.  This is achieved by introducing a system of coordinates (closely related to oblate spheroidal coordinates) in the extended manifold $\zfM$ that are {\em centered at the ring}:  Let
$$ x:= \frac{r}{a},\qquad y := \cos \theta,\qquad z:= x+iy.$$
$(x,y)$ can be thought of as a system of coordinates on the ramified cover $\Om$ of the two-dimensional quotient manifold $\cQ$.  The domain $\Om$ is  identified with a strip in the complex plane
 $$
 \Om := \{ z\in \Cset\ |\ |\mbox{Im\,}z| \leq 1\}.
 $$
 Expressing the metric $\bg_0$ of $\M_0^*$ in $(x,y,\varphi,t)$ coordinates one obtains
$$
ds_{\bg_0}^2 = -dt^2 +a^2 (1+x^2)(1-y^2) d\varphi^2 + a^2(x^2+y^2)\left(\frac{dx^2}{1+x^2} + \frac{dy^2}{1-y^2}\right)
$$
Thus the Dirichlet energy of a static axisymmetric function $f =f(x,y)$ on $\M_0^*$ is equal to
$$
E(f) = \pi a \int_{-1}^1\int_{-\infty}^\infty (1+x^2)f_x^2 + (1-y^2)f_y^2\ dx dy
$$
and a harmonic $f$ would satisfy
$$
\p_x((1+x^2)f_x) + \p_y((1-y^2)f_y) = 0.
$$
In particular the electromagnetic Ernst potential $\Phi = \phi+i\psi$ must satisfy the above equation.  Putting it in terms of the complex variable $z=x+iy$, this means that $\Phi=\Phi(z,\bar{z})$ solves the following PDE:
$$
\left(|z|^2(\p_z^2+\p_{\bar{z}}^2) + (4+z^2+\bar{z}^2)\p_z\p_{\bar{z}} + 2(z\p_{\bar{z}}+\bar{z}\p_z)\right)\Phi = 0.
$$
It is very easy to see that the only {\em meromorphic} solution to the above equation is (modulo an additive constant)
$$\Phi = \frac{C}{z}$$
for some complex constant $C$.  Similarly, the only anti-meromorphic solution is $\Phi = C/\bar{z}$.

Recall now the symmetry assumptions \refeq{sym1} and \refeq{sym2} we made in the above about the electromagnetic potentials.  In terms of the Ernst potential $\Phi$ these become the following two symmetry conditions
$$
\Phi(-z,-\bar{z}) = -\Phi(z,\bar{z}),\qquad \Phi(\bar{z},z) = \overline{\Phi(z,\bar{z})}.
$$
The first condition is satisfied for the solution we have found, and the second one holds provided the constant $C$ is real.  Writing $C=\sq/a$ and expressing the anti-meromorphic solution in terms of the original variables $(r,\theta)$ one recovers exactly the KN electrostatic potential and magnetic stream function: 
$$\phi =\frac{\sq}{a} \frac{x}{x^2+y^2} =  \frac{\sq r}{r^2+a^2\cos^2\theta},\qquad \psi = \frac{\sq}{a}\frac{y}{x^2+y^2}= \frac{a\sq\cos\theta}{r^2 + a^2\cos^2\theta}.$$\newcommand{\rr}{\textsc{r}}

We now show that each of these two is in fact unique in a much larger class:
Let
$$ \rr := \sqrt{x^2+y^2},$$ and consider the following Dirichlet problem (with prescribed singular behavior) in the open half-strip $$\Om_+ := (0,\infty) \times (-1,1)$$
\bna
\p_x((1+x^2)u_x) + \p_y((1-y^2)u_y) & =&  0 \mbox{ in } \Om_+,\label{eq:u}\\
u(0,y) & = & 0\qquad \forall y,\ 0<|y|<1\label{symu}\\
u(x,1) = u(x,-1) &= & \frac{x}{1+x^2}, \qquad \forall x>0\\
 u(x,y) & = &  \frac{1}{x} + O_1(x^{-2}) \mbox{ as } x\to \infty\label{ab}\\
u(x,y) & =& \frac{x}{\rr^2} + o_1(\frac{1}{\rr})\mbox{ as } \rr \to 0.\label{bu}
\ena
Here the subscript 1 in the notations $O_1$ and $o_1$ simply means that when $u$ is differentiated, the decay rate at infinity improves by one power, and similarly the blowup rate at zero worsens by one power.  We then have
\begin{thm}
Let $u_1$ and $u_2$ be two solutions to the above Dirichlet problem in $C^2(\Om_+)\cap C(\overline{\Om_+}\setminus\{0\})$. Then $u_1\equiv u_2$.
\end{thm}

\begin{proof}
Let $u:= u_1-u_2$. Then in $\Om_+$ the function $u$ satisfies, in addition to the equation \refeq{eq:u}, also
$$
u(0,y) =0,\ u(x,1) = u(x,-1) = 0,\ u = O_1(x^{-2}) \mbox{ as }x \to \infty,\ u = o_1(\rr^{-1})\mbox{ as } \rr\to 0.
$$
For $x>0$ and $y\ne 0$ let $$U(x,y) := \int_x^\infty u(x',y) dx'.$$
 Note that $U_x = -u$.  Let $D_h$ denote the  disk of radius $h>0$ centered at the origin of the $xy$-plane, and let $D_h^* := D_h \setminus \{(0,0)\}$ denote the punctured disk.

Let $(x,y) \in D_{1/2}^*$ with $x\geq 0$. Integrating equation \refeq{eq:u} in the variable $x$ on $[x,A]$ and letting $A\to \infty$, using the asymptotic behavior \refeq{ab}, one obtains that $U$ solves the following uniformly elliptic PDE in  the right half disk $D_{1/2}^*\cap \Om_+$:
\beq\label{eq:U}
 LU := (1+x^2)U_{xx} + ((1-y^2)U_y)_y = 0.
 \eeq
Let us extend $U$ by reflection to negative $x$, i.e. for $x<0$ set $$U(x,y) := U(-x,y).$$  The extension is then clearly a solution of \refeq{eq:U} in all of $D_{1/2}^*$. Consider also the functions
$$u_0(x,y) := \frac{x}{x^2+y^2},\qquad U_0(x,y) := - \half \ln (x^2+y^2) + \half \ln(1-y^2).$$
It's easy to check that $U_0$ satisfies $LU_0 = 0$.  Moreover, $\p_xU_0 = -u_0$,  $U_0 \geq 0$ in $D_{1/2}^*$ and $U_0 \to \infty$ as $\rr \to 0$. The following theorem is a classic result:
\begin{thm*} (Gilbarg and Serrin \cite{GilSer54}) Let $L$ be a second-order uniformly elliptic operator in a punctured disk $D_{h}^*$ in $\RR^2$.  Suppose $U_0$ is a nonnegative solution of $LU = 0$ in the punctured disk, and that $U_0\to \infty$ as $\rr \to 0$.  Then for any nonnegative solution $U$ of $LU = 0$ in the punctured disk there exists a constant $c\geq 0$ and a function $w \in C^1(D_h)$ such that
$$U = c U_0 + w.$$
\end{thm*}
Thus all we need to show is that the solution $U$ of \refeq{eq:U} thus constructed is nonnegative in some disk $D_h$.   From the above theorem it would then follow that $$u = cu_0 - w_x.$$ Multiplying this by $\rr$ and taking the limit $\rr\to 0$, using that $u=o(1/\rr)$, and that $\rr u_0$ does not tend to a limit, it then follows that $c=0$, i.e. $u= -w_x$, and thus $u$ is continuous on $D_h$.   In particular, since $u(0,y) =0$ was known for $y\ne 0$ as a result of  \refeq{symu}, it follows that $u(0,0)= 0$.   Thus we have now shown that $u=0$ on all pieces of the boundary of $\Om_+$, and moreover by \refeq{ab} we know $u\to 0$ as $x\to \infty$.  An application of the maximum principal to \refeq{eq:u} in $\Om_+$ would then allow us to conclude that $u\equiv 0$ in $\Om_+$, and thus by symmetry in all of $\Om$.

To show that $U$ is nonnegative, it suffices to prove that $U$ is bounded on one side in $D_h$.  Suppose it is not.  Then
$$\limsup_{\rr\to 0}  U = +\infty,\qquad \liminf_{\rr\to 0}U = -\infty.$$
Since $U$ is assumed continuous in the punctured disk, it follows by the intermediate value theorem that there is a sequence of points $z_m =(x_m,y_m)$ with $z_m \to 0$ as $m\to \infty$ such that $U(z_m) = 0$.  Let $D_m := D_{|z_m|}(0)$.  From the assumption \eqref{bu}  it follows that for a given $\ep>0$ there exists $M>0$ such that $\forall m >M$ and all $(x,y) \in D_m^*$,
$$|U_x(x,y)| =  |u(x,y)| \leq \frac{\ep}{\rr},\qquad |u_y(x,y)| \leq \frac{\ep}{\rr^2}.$$
Meanwhile, from \eqref{ab} we have $u_y(x,y) = O(|x|^{3})$ as $|x|\to \infty$.  It then follows that there exists a constant $C>0$ such that $$\forall m>M,\ \forall (x,y) \in D_m^*,\quad |U_y(x,y)| \leq \int_x^\infty |u_y(x',y)| dx' \leq \frac{C}{r}.$$
Therefore $$|\nab U| \leq \frac{C}{\rr}$$ i.e. the gradient of $U$ is uniformly bounded on each circle.  Let $z\in \Om$ be such that $|z|=|z_m|$.  Then if $A(z_m,z)$ denotes the arc of the circle from $z_m$ to $z$, we have
$$
|U(z)| = |U(z) - U(z_m)| = | \int_{A(z_m,z)} \p_\theta U d\theta | \leq 2\pi C.
$$
I.e. $U$ is bounded uniformly on each circle.  But then by the maximum principle,  we can conclude that $U$ is bounded in a small enough disk $D_h$, contradicting the assumption that it was unbounded on both sides.  It thus follows that either $U\leq K$ or $U\geq -K$ holds in some $D_h$, for some constant $K>0$.  Suppose without loss of generality that it is the latter case: $U+K\geq 0$.  Define $\tilde{U} = U+K$.  Then $\tilde U$ is a nonnegative solution of $LU= 0$ in $D_h$, to which the Gilbarg-Serrin result can be applied.
\end{proof}

Similarly, we can prove the uniqueness of the KN magnetic stream function $\psi$. 
\begin{thm}

 Let $$\Om_0 := \RR \times (0,1)$$ and consider the Dirichlet problem
\bna
((1+x^2)v_x)_x + ((1-y^2)v_y)_y & = & 0 \mbox { in }\Om_0 \label{eq:v}\\
v(x,0) & = & 0\quad \forall x\ne 0\label{vic} \\
v(x,1) & = & \frac{1}{1+x^2}\quad \forall x \in \RR\\
v(x,y) & = &  \frac{y}{x^2} + O_1(x^{-3}) \mbox{ as } |x|\to \infty\label{vab}\\
v(x,y) & = & \frac{y}{\rr^2} + o_1(\frac{1}{\rr})\mbox{ as } \rr \to 0.\label{vbu}
\ena
Then $$v_0(x,y) := \frac{y}{\rr^2}$$ is the only solution of this Dirichlet problem in $C^2(\Om_0)$.  
\end{thm}
\begin{proof}  Let $v := v_1 - v_2$ be the difference of two solutions of the above.  Then in addition to the equation \refeq{eq:v}, $v$ also satisfies
 $$
 v(x,0) = 0,\quad v(x,1) = 0,\quad v(x,y) = O_1(x^{-3})\mbox{ as }x\to \infty,\quad v(x,y) = o_1(\rr^{-1})\mbox{ as }\rr\to 0.
 $$
  By analogy to the electric potential case discussed in the above, for $(x,y) \in D_{1/2}^*\cap \{ y\geq 0\}$ let $$V(x,y) := \int_y^1 v(x,y') dy'.$$
Integrating now \refeq{eq:v} in the variable $y$ on $[y,1]$ we obtain that $V$ is a solution of
$$ L'V := ((1+x^2)V_x)_x +(1-y^2)V_{yy} = 0.$$
We then extend $V$ by reflection to all of $D_{1/2}^*$, i.e. by setting, for $y<0$, $$V(x,-y) := -V(x,y).$$
The extended $V$ is then a solution of $L'V=0$ in the punctured disk $D_{1/2}^*$.  Defining
$$V_0(x,y) := -\half \ln(x^2+y^2) + \half \ln(1+x^2),$$
one sees that $L'V_0 = 0$, $V_0\geq 0$ in $D_{1/2}^*$, and $V_0 \to \infty$ as $\rr \to 0$.  Thus once again by the Gilbarg-Serrin result, for any other solution of $L'V= 0$ that is nonnegative in $D_h^*$ for some $h>0$, there exists $c\geq 0$  and $w' \in C^1(D_h^*)$ such that
\beq V = cV_0 + w'.\label{VV0}\eeq
Proof of nonnegativity of $V$ in $D_h^*$ is the same as that of $U$ in the above.
Differentiating \refeq{VV0} in $y$, we have $v = cv_0 - w'_y$ and once again multiplying this by $\rr$ and taking the limit $\rr\to 0$ we obtain that $c=0$ and that $v$ is continuous in $D_h$.  Thus $v$ is a classical solution of \refeq{eq:v} with zero Dirichlet data and zero asymptotic value, so that by maximum principle $v\equiv 0$.
\end{proof}

We conclude by noting that any constant multiple of $\psi$ would also be a solution of \refeq{eq:v} and a legitimate magnetic stream function. If we do not scale $\phi$ in the same way, all that happens is that  $\phi$ and $\psi$  will no longer be the real and imaginary parts of the same meromorphic function, but otherwise they represent physically meaningful electric and magnetic fields that may still be viewed as coming from the sources with the same singular support but with different strengths.  Hence, one can come up with a  generalization of the zGKN family by considering, on the same background (which solves the Einstein vacuum equations) and for any pair of numbers $(\sq,\textsc{i})$ the potentials
$$
\phi = \frac{\sq r}{r^2+ a^2\cos^2\theta},\qquad \psi =  \frac{\textsc{i}\pi a^2\cos\theta}{r^2  + a^2\cos^2\theta}
$$
where $\textsc{i}$ is the  electric current in a loop of radius $a$ (in a single-sheeted space)  that would produce a magnetic dipole of the same strength as  the KN field.  Since in the $G\to 0$ limit Einstein's equations decouple from Maxwell's, the zero-$G$ Kerr manifold, which is a solution of Einstein's vacuum equations, together with the  electromagnetic field produced by the above potentials will be a solution of the Einstein--Maxwell system.  This generalization will be further studied  in the companion paper \cite{KieTahIa}.

\section{Summary and outlook}
The results of this paper come closest to answering the question ``What is the source of a KN metric?"   
The interior of the KN space-time with $G>0$ is too pathological for one to be able to speak of ``spinning sources" in any conventional sense that has been given to these words so far.   However,  in a zero-gravity
$G\to0$ limit the pathologies go away, leaving behind no detectable geometric trace, but only  a topological remnant in the form of
multi-sheetedness of the spacetime. 
We have studied Maxwell's equations on this background, and have given an {\em ab initio} derivation of the KN fields as stationary solutions coming from charge and current distributions that are supported only on the ring singularity of the two-sheeted space. We have analyzed the behavior of the KN fields near the ring and have  proven that their singularity corresponds neither to a monopole nor a dipole, but something in between, ie. it is {\em sesqui-polar}.  We have proved that the KN potentials are unique in the class of all potentials with the same prescribed blow-up rate at the singular ring.  

In the companion paper \cite{KieTahIa} we embark on a thorough study of the Dirac equation on a zGKN background, and show how the double-sheeted nature of the space is reflected in the surprising spectral properties of the Dirac Hamiltonian on this background.

\bigskip 

\noindent{\bf Acknowledgements. } The author is extremely grateful to Michael Kiessling for being a constant source of ideas and encouragement throughout the years, and for many hours of fruitful discussion on this particular project.

\small
\baselineskip=13pt
\bibliographystyle{plain}

\begin{thebibliography}{10}
\bibitem{Appell}
P.~Appell.
\newblock Quelques remarques sur la th\'eorie des potentiels multiforms.
\newblock {\em Math. Ann.}, 30:155--156, 1887.

\bibitem{ADM61}
R.~Arnowitt, S.~Deser, and C.~W. Misner.
\newblock Coordinate invariance and energy expressions in general relativity.
\newblock {\em Phys. Rev.}, 122:997--1006, 1961.

\bibitem{BacWey22}
R.~Bach and H.~Weyl.
\newblock Neue {L}\"osungen der {E}insteinschen {G}ravitationsgleichungen.
\newblock {\em Math. Zeit.}, 13:134--145, 1922.

\bibitem{Bic93}
J.~Bi\v{c}\'ak and T.~Ledvinka.
\newblock Relativistic disks as sources of the {K}err metric.
\newblock {\em Phys. Rev. Lett.}, 71:1669--1672, Sep 1993.

\bibitem{BonSac68}
W.~B~Bonnor and A.~Sackfield.
\newblock The interpretation of some spheroidal metrics.
\newblock {\em Comm. Math. Phys.}, 8(4):338--344, 1968.

\bibitem{BoyLin67}
R.~H.~Boyer and R.~W.~Lindquist.
\newblock Maximal analytic extension of the {K}err metric.
\newblock {\em Jour.  Math. Phys.}, 8(2):265--281, 1967.

\bibitem{Car68}
B.~Carter.
\newblock Global structure of the {Kerr} family of gravitational fields.
\newblock {\em Phys. Rev.}, 174:1559--1571, 1968.

\bibitem{Cha83}
S.~Chandrasekhar.
\newblock {\em The Mathematical Theory of Black Holes}.
\newblock Oxford University Press, New York, 1983.

\bibitem{Ernst}
F.~J.~Ernst.
\newblock New formulation of the axially symmetric gravitational field problem.
\newblock {\em Phys. Rev.}, 167:1175--1178, Mar 1968.

\bibitem{Eva51}
G.~C.~Evans.
\newblock {\em Lectures on Multiple Valued Harmonic Functions in Space}.
\newblock Univ. of California Press, Berkeley and Los Angeles, 1951.

\bibitem{GaiLyn07}
J.~R.~Gair and D.~Lynden-Bell.
\newblock Electromagnetic fields of separable spacetimes.
\newblock {\em Class. Quant. Grav.}, 24(6):1557, 2007.

\bibitem{GilSer54}
D.~Gilbarg and J.~Serrin.
\newblock On isolated singularities of solutions of second order elliptic
  differential equations.
\newblock {\em Jour. d'Anal. Math.}, 4:309--340, 1954.

\bibitem{HawEll73}
S.~Hawking and G.~Ellis.
\newblock {\em The Large Scale Structure of Space-Time}.
\newblock Cambridge Univ. Press, Cambridge, 1973.

\bibitem{Isr70}
W.~Israel.
\newblock Source of the {K}err metric.
\newblock {\em Phys. Rev. D (3)}, 2:641--646, 1970.

\bibitem{Kai04}
G.~Kaiser.
\newblock Distributional sources for {N}ewman's holomorphic {C}oulomb field.
\newblock {\em Jour. Phys. A}, 37(36):8735,
  2004.

\bibitem{KieTahIa}
M.~K.-H. Kiessling and A.~S. Tahvildar-Zadeh.
\newblock The Dirac point electron in zero-gravity Kerr--Newman spacetime.
\newblock {\em preprint}, pages 1--48, 2014.

\bibitem{Kle03}
C.~Klein.
\newblock On explicit solutions to the stationary axisymmetric
  {E}instein-{M}axwell equations describing dust disks.
\newblock {\em Ann. Phys.}, 12(10):599--639, 2003.

\bibitem{Lew32}
T.~Lewis.
\newblock Some special solutions of the equations of axially symmetric
  gravitational fields.
\newblock {\em Proc. Roy. Soc. Lon.  A}, 136(829):pp. 176--192,
  1932.

\bibitem{Lop83}
C.~Lopez.
\newblock Material and electromagnetic sources of the {K}err-{N}ewman geometry.
\newblock {\em Il Nuovo Cimento B (1971-1996)}, 76:9--27, 1983.

\bibitem{Lyn04}
D.~Lynden-Bell.
\newblock Electromagnetic magic: The relativistically rotating disk.
\newblock {\em Phys. Rev. D}, 70:105017, Nov 2004.

\bibitem{Lyn04a}
D.~Lynden-Bell.
\newblock Relativistically spinning charged sphere.
\newblock {\em Phys. Rev. D}, 70:104021, Nov 2004.

\bibitem{Maz83}
P.~O. Mazur.
\newblock A relationship between the electrovacuum {E}rnst equations and
  nonlinear {$\sigma $}-model.
\newblock {\em Acta Phys. Polon. B}, 14(4):219--234, 1983.

\bibitem{NCCEPT65}
E.~T. Newman, E.~Couch, K.~Chinnapared, A.~Exton, A.~Prakash, and R.~Torrence.
\newblock Metric of a rotating, charged mass.
\newblock {\em Jour. Math. Phys.}, 6(6):918--919, 1965.

\bibitem{NewJan65}
E.~T. Newman and A.~I. Janis.
\newblock Note on the {K}err spinning-particle metric.
\newblock {\em Jour. Math. Phys.}, 6:915--917, 1965.

\bibitem{New02}
E.~T. Newman.
\newblock Classical, geometric origin of magnetic moments, spin-angular
  momentum, and the {D}irac gyromagnetic ratio.
\newblock {\em Phys. Rev. D}, 65:104005, Apr 2002.

\bibitem{Pap53}
A.~{Papapetrou}.
\newblock {Eine rotationssymmetrische {L}{\"o}sung in der allgemeinen
  {R}elativit{\"a}tstheorie}.
\newblock {\em Ann. Phys.}, 447:309--315, 1953.

\bibitem{PekFra87}
C.~L. Pekeris and K.~Frankowski.
\newblock The electromagnetic field of a {K}err-{N}ewman source.
\newblock {\em Phys. Rev. A (3)}, 36(11):5118--5124, 1987.

\bibitem{Som96}
A.~Sommerfeld.
\newblock Mathematische {T}heorie der {D}iffraction.
\newblock {\em Math. Ann.}, 47(2-3):317--374, 1896.

\bibitem{Som97}
A.~Sommerfeld.
\newblock {\"U}ber verzweigte {P}otentiale im {R}aum.
\newblock {\em Proc. London Math. Soc.},
  s1-28(1):395--429, 1896.

\bibitem{Tah11}
A.~S.~Tahvildar-Zadeh.
\newblock On the static spacetime of a single point charge.
\newblock {\em Rev. Math. Phys.}, 23(3):309--346, 2011.

\bibitem{Tio73}
J.~Tiomno.
\newblock Electromagnetic field of rotating charged bodies.
\newblock {\em Phys. Rev. D}, 7:992--997, Feb 1973.

\bibitem{Wei90}
G.~Weinstein.
\newblock On rotating black holes in equilibrium in general relativity.
\newblock {\em Comm. Pure Appl. Math.}, 43(7):903--948, 1990.

\bibitem{Wei96}
G.~Weinstein.
\newblock {$N$}-black hole stationary and axially symmetric solutions of the
  {E}instein/{M}axwell equations.
\newblock {\em Comm. Par. Diff. Eq.}, 21(9-10):1389--1430,
  1996.

\bibitem{Wey17}
H.~Weyl.
\newblock Zur gravitationstheorie.
\newblock {\em Ann. Phys. (Germany)}, 54:117, 1917.

\bibitem{Zip66}
D.~M. Zipoy.
\newblock Topology of some spheroidal metrics.
\newblock {\em Jour. Math. Phys.}, 7(6):1137--1143, 1966.

\end{thebibliography}
\def\cprime{$'$} \def\cprime{$'$}

\end{document}